\documentclass[12pt]{amsart}
\usepackage{amsmath}
\usepackage{amsfonts}
\usepackage{amssymb}
\usepackage{mathrsfs}
\usepackage{latexsym}
\usepackage{paralist}
\usepackage{color}
\usepackage{url}
\usepackage{enumitem}
\usepackage{bbold}
\usepackage[foot]{amsaddr}
\usepackage{orcidlink}

\usepackage[OT2,OT1]{fontenc}
\newcommand\cyr{%
	\renewcommand\rmdefault{wncyr}%
	\renewcommand\sfdefault{wncyss}%
	\renewcommand\encodingdefault{OT2}%
	\normalfont
	\selectfont}
\DeclareTextFontCommand{\textcyr}{\cyr} 

\def\bP {\mathbf{P}}

\def\cD {\mathcal{D}}

\def\cM {\mathcal{M}}

\def\de {{\delta}}
\def\eps {{\epsilon}}

\def\ka {{\kappa}}
\def\l {{\lambda}}

\def\si {{\sigma}}

\def\pa {{\partial}}
\def\grad {{\nabla}}

\def\indc {{\bf 1}}

\def\deb {{\rightharpoonup}}

\def\be {\begin{equation}}
\def\ee {\end{equation}}

\newcommand{\ba}{\begin{aligned}}
	\newcommand{\ea}{\end{aligned}}

\newcommand{\Div}{\operatorname{div}}

\newcommand{\Supp}{\operatorname{supp}}

\newcommand{\Lip}{\operatorname{Lip}}

\newtheorem{thm}{Theorem}

\newtheorem{lem}[thm]{Lemma}
\newtheorem{prop}[thm]{Proposition}

\newtheorem{remark}{Remark}

\def\var{\varepsilon}

\def\R{{\Bbb R}}
\def\E{{\Bbb E}}
\def\N{\Bbb N}

\def\blackbox{\unskip\kern 6pt\penalty 500%
\raise -1pt\hbox{\vrule\vbox to 8pt{\hrule width 6pt\vfill\hrule}\vrule}}

\begin{document}

\title[slowing particles]{A model for slowing particles in random media}

\author{Fran\c{c}ois Golse$^{1}$\,\,\orcidlink{0000-0002-7715-6682}}
\address{$^1$École polytechnique \& IP Paris\\ CMLS\\ 91128 Palaiseau Cedex, France}
\email{francois.golse@polytechnique.edu}

\author{Valeria Ricci$^{2,\ast}$\,\,\orcidlink{0000-0002-3902-905X}}
\address{$^2$Universit\`a degli Studi di Palermo\\ 
	Dipartimento di Matematica e Informatica\\
	Via Archirafi 34\\
	90123 Palermo, Italy
	}
\email{valeria.ricci@unipa.it}

\author{Ana Jacinta Soares$^{3}$\,\,\orcidlink{0000-0003-4771-9859}}
\address{$^3$Universidade do Minho\\
	Centro de Matem\'atica\\
	Campus de Gualtar\\
	4710-057 Braga, Portugal\\
}
\email{ajsoares@math.uminho.pt}

\address{$^{\ast}$Corresponding author}

\thanks{\textbf{Acknowledgments}
V.\,Ricci acknowledges the support of the Gruppo Nazionale per la Fisica Matematica-Istituto Nazionale di Alta Matematica (GNFM -\newline INdAM) and of  the University of Palermo, Fondi Finalizzati alla Ricerca di Ateneo FFR2021 Valeria Ricci, FFR2023, FFR2018-2020.\\ A.J.\,Soares thanks the support of the Portuguese FCT (Funda\c c\~ao para a Ci\^encia e a Tecnologia) Projects UIDB/00013/2020 and UIDP/00013/2020 of CMAT-UM (Centro de Matem\'atica da Universidade do Minho).\\
F. Golse and A.J. Soares thank the support of the Bilateral Project Pessoa 
PT-FCT-Ref.2021.09255.CBM and FR-PHC-Ref.47871WF. }

\begin{abstract}
	
	We present a simple model in dimension $d\geq 2$ for slowing particles in random media, where point particles move in straight lines among and inside spherical identical obstacles with Poisson distributed centres. When crossing an obstacle, a particle is slowed down according to the law $\dot{V}= -\frac{\ka}{\eps} S(|V|) V$, where $V$ is the velocity of the point particle, $\ka$ is a positive constant, $\eps$ is the radius of the obstacle and $S(|V|)$ is a given slowing profile. With this choice, the slowing rate in the obstacles is such that the variation of speed at each crossing is of order $1$. We study the asymptotic limit of the particle system when $\eps$ vanishes and the mean free path of the point particles stays finite. We prove the convergence of the point particles density measure to the solution of a kinetic-like equation with a collision term which includes a contribution proportional to a $\de$ function in $v=0$; this contribution guarantees the conservation of mass for the limit equation. 

\end{abstract}

\maketitle



\section{Introduction} \label{intro}

In this article, we present a simple particle model describing the slowing down of point (\textit{light}) particles in a random inhomogeneous medium, where the inhomogeneities (or \textit{inclusions}) consist in spherical obstacles with Poisson distributed centres. In this model, the light particles move freely in the space 
free from obstacles and slow down when crossing the obstacles. We study the asymptotic behaviour of the particle system when the radius $\eps$ of the obstacles vanishes and the slowing rate grows to infinity in such a way that the loss of kinetic energy of the particles at each crossing is of order $1$. Choosing 
the density of the Poisson centres in such a way that the mean free path is finite, we derive a kinetic type equation for the probability measure of the light particles in the phase space where the conservation of mass is guaranteed by a term proportional to a Dirac delta measure in $v=0$.

One motivation for studying this type of dynamics comes from the modelling of Inertial Confinement Fusion (ICF) targets. While a detailed description of the intricate processes at work in ICF targets \cite{Duderstadt, AtzeniMeyer} exceeds 
by far the scope of the present study, the following features are of special interest to the problem addressed in the present paper. 

Some ICF experiments use pure deuterium (D) targets as a mean of measuring the fuel areal density, known to be a key parameter in ICF research \cite{Azechi,Cable}. In these targets, tritium (T) ions are created by the
thermonuclear reaction 
\[
D+D\to T+\text{ proton }+4.03MeV\,.
\]
These T ions are created with a very large kinetic energy (about 1 MeV, see Section 2.1.1 in \cite{Duderstadt} or Section 1.3.1 in \cite{AtzeniMeyer}), and therefore are referred to as \textit{suprathermal}. 
The T ions travel in the plasma and lose kinetic energy, to 
the point that they combine with the D ions, producing $\alpha$ particles, neutrons and energy according to the thermonuclear reaction
\[
D+T\to\alpha+\text{ neutron }+17.58MeV\,.
\]
The cross-section of the D-T reaction is much larger than that of the D-D reaction, and the D-T reaction is the main source of energy in ICF experiments. Thus, the yield of ICF implosions will be reduced by any obstruction 
to the D-T reaction.

One significant such obstruction comes from the ICF implosion itself. The fuel is highly compressed by ablation of the outer shells of the fusion pellet (see Section 3.8 in \cite{Duderstadt}). This process initiates a 
Richtmyer-Meshkov instability (see Section 3.6 in \cite{Duderstadt}, Section 8.1.3 in \cite{AtzeniMeyer} and \cite{Richt}) at the interface between the fuel and the various layers of materials around it. As a result, inclusions of dense, inert material with a high 
stopping power are driven in the fuel. T ions stopping in these inclusions are prevented from reacting with the surrounding D ions, with the negative effect on the yield of the thermonuclear burning of the fuel in the target 
already mentioned above.

Returning to the mathematical problem sketched at the beginning of this section, the point particles in the context of ICF experiments would be the T ions, and the ``light'' material would be the thermonuclear fuel, while
the inclusion of ``dense'' material with high stopping power would be the debris of the various outer shells surrounding the fusion pellet before implosion driven into the fuel by the hydrodynamic instabilities.

The specific problem mentioned above has been addressed in part in earlier mathematical publications. For instance Levermore and Zimmerman have obtained a formula for the fraction of T ions lost to the D-T reaction 
(formula (22) in \cite{LevermoreZimmerman}). Their argument is based on an earlier model of particle transport in a very special kind of random media proposed by Levermore, Pomraning, Sanzo and Wong 
\cite{LevermorePomraningSanzo}. In such media, a particle in medium $A$ (resp. $B$)  at the position $s$ has a probability of finding itself in medium $B$ (resp. $A$) at the position $s+ds$ equal to $ds/\lambda_A$ (resp. 
$ds/\lambda_B$), where the lengths $\lambda_A$ and $\lambda_B$ are constant. The same problem is considered by Clouet, Puel, Sentis and the first author in \cite{ClouetGolsePuelSentis} with a completely different
random media, corresponding to the setting chosen in the present paper --- spherical inclusions centred at realizations of a Poisson point process. The analysis in that reference leads to another formula for the
probability of stopping in the inclusions (see formula (13) in \cite{ClouetGolsePuelSentis}). Numerical simulations lead to comparisons between the results obtained in  \cite{LevermoreZimmerman} and \cite{ClouetGolsePuelSentis}
concerning the fraction of T ions lost to the D-T reaction.

The present paper differs from either \cite{LevermoreZimmerman} or \cite{ClouetGolsePuelSentis} in several ways. The random medium considered here involves a Poisson point process as in \cite{ClouetGolsePuelSentis}, but
corresponds to a different asymptotic regime. Specifically, the situation considered here involves 

\begin{enumerate}[label=(\roman*)]
	\item very small, identical spherical inclusions,
	
\item  very high stopping power in each spherical inclusion.
\end{enumerate}

We shall see that the effective equation for the expected velocity distribution function of point particles (T ions), equation (\ref{eqlim2}) in Section \ref{matdef},
is a kinetic model including a collision integral acting 
on the particles speed only, and not on their direction as  in most classical kinetic models, together with a Lagrange multiplier, equation (\ref{moltlag}) in Section \ref{matdef}, proportional to the Dirac mass at zero velocity, which 
maintains the local conservation of mass.

The outline of the paper is the following:  in Section \ref{matdef} we present the particle model, fix the notation and state the main theorem; in Section \ref{apriori} we recall some a priori estimates needed to derive some useful properties of the limit points of the probability measure of light particles; in Section \ref{Equalim} we derive the limit equation for the probability measure of light particles and prove the main theorem.

\section{The particle model and the main result}
\label{matdef}

In this section we shall define the particle model and state the main theorem.

\bigskip

In what follows, we denote 
by $$B(c,\eps)=\{x\in\R^d: |x-c|<\eps\}$$ the open ball centred in $c$ and having radius $\eps$ (we shall add a superscript denoting the dimension when needed, so that $B^d(c,\eps)$ will denote the open ball centred in $c$ and having radius $\eps$ in dimension $d$) and by $B^d=|B^d(0,1)|$ the volume of the $d$-dimensional unit ball.

Moreover, given a locally compact Hausdorff space $X$, we denote by $\cM(X)$ the space of Radon measures, by $\cM^1(X)$ the space of finite mass signed Radon measures and by $\cM^1_+(X)$ the space of positive finite mass Radon measures on $X$.

In this paper we shall use the following notions of convergence in Radon measures spaces. 

Denoting as $C_c(X)$
and $C_b(X)$ resp. the set of continuous functions with compact support
and the set of continuous bounded functions on $X$, we say that a sequence of measures on $X$, $\{\mu_n\}$, converges to $\mu$ in the 

\begin{itemize}
	\item {\bf weak-*} (of $\cM(X)$) or {\bf{vague topology}} when, for all $\phi \in C_c(X)$,  $\int_{X} \phi(x)\mu_n(dx)\to\int_{X} \phi(x)\mu(dx)$,  and we write:
	$$
	\mu_n \stackrel{*}{\deb} \mu \quad \mathrm{as}\quad n\to\infty;
	$$
	
	\item {{\bf weak topology}}  (of  $\cM^1(X)$)  when, for all $\phi \in C_b(X)$, \newline $\int_{X} \phi(x)\mu_n(dx)\to\int_{X} \phi(x)\mu(dx)$,  and we write:
	$$
	\mu_n \deb \mu \quad \mathrm{as}\quad n\to\infty.
	$$
\end{itemize} 

As a last remark on the notation, if $m(dxdy)$ is a positive Radon measure with finite mass on $\R^p\times\R^q$, we denote by
\[
\int_{\R^q}mdy
\]
the push-forward of $m$ by the first projection $\R^p\times\R^q\ni(x,y)\mapsto x\in\R^p$, i.e. the Radon measure defined on $\R^p$ by the linear functional
\[
\psi\mapsto\int_{\R^p\times\R^q}\psi(x)m(dxdy)\,.
\]
We use a similar notation for the push-forward of $m$ by the second projection $\R^p\times\R^q\ni(x,y)\mapsto y\in\R^q$.

\bigskip

We study then the asymptotics of the following particle system.

\smallskip

We consider a particle with initial position $x\in\R^d$ and velocity $v\in\R^d$, $d\geq 2$, and configurations of obstacles (or inclusions) of radius $\epsilon$ with Poisson distributed centres $C=(c_1,c_2\ldots)$, where $c_i\in\R^d$.

The probability to find $n$ centres in the set  $A\subset\R^d$ is then given by:
$$
\bP\left(\{C\text{ s.t. }\#(A\cap C)=n\}\right)=\frac{(\l_{\eps}|A|)^n}{n!}e^{-\l_{\eps}|A|},
$$
with $\l_{\eps}=\frac{\l}{\eps^{d-1}}$ and $\l=O(1)$ ($|A|$ denotes the Lebesgue measure of the set $A$). Hence the particle mean free path is finite uniformly as $\eps\to 0^+ $, since the mean free length to the first collision is proportional to $(\l_\eps \eps^{d-1})^{-1}$ and the particle is slowed down inside the obstacles.

We shall denote by $\E$ the expectation under $\bP$  and by $$Z_\eps[C]:=\R^d\setminus\bigcup_{c\in C}\overline{B(c,\eps)}$$ the complement in $\R^d$ of the set occupied by the obstacles.

\bigskip

The equation of motion of the particle for a given configuration of obstacles $C$ is given by (see also \cite{LevermoreZimmerman,ClouetGolsePuelSentis}):

\be
	\label{moto}
	\left\{
	\begin{array}{l}
		\dot X=V,\\
		\\
		\dot V=-\frac{\ka}\eps\indc_{Z^c_\eps[C]}(x)S(|V|)V=-\ka_\eps(x;C)S(|V|)V,\\
		\\
		(X,V)\big|_{t=0}=(x,v),
	\end{array}
	\right.
\ee
where  $S\in \Lip((0,+\infty))$, satisfying $S\geq S_0>0$, is the slowing profile ($\Lip((0,+\infty))$ denotes the space of Lipschitz functions on $(0,+\infty)$).

Calling $(X^\eps_t,V^\eps_t)(x,v;C)$ the solution to the differential system (\ref{moto}), the evolution of an initial profile both on the space of test functions  and on the space of Borel measures is defined through the following prescriptions:
\[
\label{defevotest}
\phi_\eps(t,x,v;C):=\phi((X^\eps_t,V^\eps_t)(x,v;C)),
\]
and
\be
\label{defevomis}
\int f_\eps(t,dxdv;C)\phi(x,v):=\int f^0_\eps(dxdv)\phi_\eps(t,x,v;C),
\ee
for 
each $\phi\in C_b (\R^d\times\R^d)$ and for each positive Borel measure with finite mass $f_\eps(0,dxdv;C)=f^0_\eps(dxdv)$.

The time-dependent measure $f_\eps$ is the weak solution 
to the Cauchy problem for the advection equation
\be
\label{evomisura}
(\pa_t+v\cdot\grad_x)f_\eps(t,dxdv;C)=\ka_\eps(x;C)\Div_v(S(|v|)vf_\eps(t,dxdv;C))
\ee
 with initial datum 
 $f_\eps(0,dxdv;C)=f^0_\eps(dxdv)$.
 
 Observe that $f_\eps$ is continuous in $t$ with values in the set of positive measures with finite mass on $\R^d\times\R^d$ equipped with the weak topology (see Lemma 
 8.1.6 in \cite{Ambrosio}).

 We are interested in studying the behaviour in the limit $\eps \to 0$ of $\E[f_\eps(t,dxdv;C)]$, i.e. in studying the Boltzmann--Grad limit for this specific particle system.
 
 \bigskip
 
 In what follows we shall prove the following theorem:
 
 \bigskip
 
\begin{thm}
\label{TeorPrinc}
Assume $f_0\in C(\R^d\times\R^d)$, with compact support and $f_0\geq 0$, 
and let $f_\eps (t,x,v;C)$ be the solution to equation (\ref{evomisura}) with initial datum $f_0$.		
 	Let  
 	$
 	a(z)=\int_{0}^{z} du \frac{1}{S(u)}
 	$, with $S\in \Lip((0,+\infty))$, satisfying $S\geq S_0>0$.
 	\smallskip
 	
 	Then 
	$\E[f_\eps(t,dxdv;C)]dt\deb F(t,dxdv)dt$ where $F(t,dxdv)$ is the unique solution in $\cD'((0,+\infty)\times\R^d\times\R^d)$ to the equation	
	\begin{eqnarray} 	
	\pa_t F + v\cdot\grad_x F  \!\!&=&\!\! \si \bigg[ \int_{0}^{1}dh^{d-1} \frac{|z|^d}{|v|^{d-1}}\frac{S(|z|)}{S(|v|)}F(t,x,z) \indc_{|v|>0} 
	\nonumber
	\\[2mm]
	&-& |v| F(t,x,v) \bigg] dxdv  +\l_F(t,x)\de_0(dv)dx
	\label{eqlim2}
 	\end{eqnarray}
 	with $|z|= a^{-1}(a(|v|)+2 \ka\sqrt{1-h^2})$, $z=|z|\hat{v}$, $\si=\lambda B^{d-1}$
 	and
 	\be
 	\label{moltlag}
 	\l_F(t,x)=\int_{\R^d} dz \int_{0}^{1} dh^{d-1}|z| F(t,x, z)\indc_{|z|\leq a^{-1}(2\ka\sqrt{1-h^2})} .
 	\ee
	
 \end{thm}

 \section{A priori estimates and considerations about limit points} \label{apriori}
 
 Before analysing in detail the evolution induced on the density measure by (\ref{moto}), we shall recall some a priori estimates which can be directly derived from the definitions (\ref{moto}), (\ref{defevomis})  and the evolution equation (\ref{evomisura}) implied by them.

   The a priori estimates worked out in this section will help us to obtain some useful properties of the limit points of $\E[f_\eps(t,x,v;C)]$ when $\eps\to 0^+$ in both these topologies.

   \begin{lem}
   	\label{cor}
   	Assume that $f_\eps^0(\cdot;C)$ is supported in $\overline{B(0,R)}\times\overline{B(0,R)}$ for some $R>0$ and all configuration $C$ of inclusions. Then, for all $t\ge 0$,
   	\be
   		\label{supptxv}
   	\mathrm{supp}(f_\eps(t,\cdot;C))\subset\overline{B(0,R+Rt)}\times\overline{B(0,R)}\,.
   	\ee
   \end{lem}
   
   \begin{proof}
   	Let $\phi\in C_b(\R^d\times\R^d)$ satisfy $\phi=0$ on $\overline{B(0,R+RT)}\times\overline{B(0,R)}$ for some $T>0$. Since
   	\[
   	\tfrac{d}{dt}|V^\eps_t(x,v;C)|=-\kappa_\eps(X^\eps_t(x,v;C),C)S(|V^\eps_t(x,v;C)|)|V^\eps_t(x,v;C)|\le 0\,,
   	\]
   	\[
   	\tfrac{d}{dt}|X^\eps_t(x,v;C)|\le|V^\eps_t(x,v;C)|\le |v|\;.
   	\]
   	Therefore $\phi((X^\eps_t,V^\eps_t)(x,v;C))=0$ for all $(x,v)\in\overline{B(0,R)}\times\overline{B(0,R)}$ and for all $t\in[0,T]$, so that
   	\begin{align*}
        \int_{\R^d\times\R^d} & \phi(x,v) f_\eps(t,dxdv;C)&
   		\\
   		&=\int_{\overline{B(0,R)}\times\overline{B(0,R)}}\phi((X^\eps_t,V^\eps_t)(x,v;C))f_\eps^0(dxdv;C)
		=0\,.
   	\end{align*}   	Therefore $\mathrm{supp}(f_\eps(t,\cdot;C))\subset\overline{B(0,R+RT)}\times\overline{B(0,R)}$ for all $t\in[0,T)$, which implies the inclusion in the lemma.
   \end{proof}

   \begin{prop}
   	\label{apriori1}
   	Assume that $f_\eps^0(\cdot;C)$ is supported in $\overline{B(0,R)}\times\overline{B(0,R)}$ for some $R>0$ and all configuration $C$ of inclusions. Then $f_\eps$
   	satisfies
   	
   	\noindent
   	\begin{enumerate}[label=(\alph*)]
   		\item \label{cont} the continuity equation
   		\[
   		\partial_t\int_{\R^d}f_\eps(t,\cdot;C)dv+\nabla_x\cdot\int_{\R^d}f_\eps(t,\cdot;C)vdv=0\quad\text{ in }\mathcal D'((0,+\infty)\times\R^d)\,;
   		\]
   		\item \label{mass}
   		the global conservation of mass
   		\[
   		\frac{d}{dt}\int_{\R^d\times\R^d}f_\eps(t,dxdv;C)=0\quad\text{ in }\mathcal D'((0,+\infty))\,;
   		\]
   		\item \label{mom}
   		the momentum balance
   		\[
   		\begin{aligned}
   			\frac{d}{dt}\int_{\R^d\times\R^d}vf_\eps(t,dxdv;C)=-\int_{\R^d\times\R^d}\kappa_\eps(x,C)S(|v|)vf_\eps(t,dxdv;C)
   			\\
   			\text{ in }\mathcal D'((0,+\infty))&\,;
   		\end{aligned}
   		\]
   		\item \label{funz}
   		the mass flux inequality
   		\[
   		\begin{aligned}
   			\frac{d}{dt}\int_{\R^d\times\R^d}\phi(|v|^2)f_\eps(t,dxdv;C)\le 0\quad\text{ in }\mathcal D'((0,+\infty))\,,
   		\end{aligned}
   		\]
   		provided that $\phi\in C^1([0,+\infty))$ is nondecreasing.
   	\end{enumerate}
   	
   \end{prop}
   
   \begin{proof}
   	Since $f_\eps(t,\cdot;C)$ is compactly supported by Lemma \ref{cor}, both sides of the advection equation (\ref{evomisura}) for $f_\eps$ are compactly supported distributions for $t\in(0,T)$, and can
   	therefore be evaluated on $C^\infty$ functions, not necessarily with compact support. Choosing the test function of the form $(t,x,v)\mapsto\psi(t,x)$ with $\psi$ of class $C^\infty$
   	on $(0,+\infty)\times\R^d$ with support in $(1/n,n)\times\R^d$ for $n\ge 1$ shows that
   	\[
   	-\int_{(0,+\infty)\times\R^d\times\R^d}(\partial_t+v\cdot\nabla_x)\psi(t,x)f_\eps(t,dxdv;C)dt=0\,.
   	\]
   	The validity of this equality for all such test function $\psi$ and all $n\ge 1$ is exactly the formulation of \ref{cont} in the sense of distributions on $(0,+\infty)\times\R^d$. The same
   	argument with $\psi$ independent of $x$ implies \ref{mass}.
   	
   	If we pick a test vector field of the form $(t,x,v)\mapsto\chi(t)v$ where $\chi\in C^\infty_c((0,+\infty))$, we find that
   	\begin{align*}
   	\int_{(0,+\infty)\times\R^d\times\R^d} & \chi'(t)vf_\eps(t,dxdv;C)dt
   		\nonumber \\[2mm]
   		& = \! \int_{(0,+\infty)\times\R^d\times\R^d}\chi(t)\kappa_\eps(x,C)S(|v|)vf_\eps(t,dxdv;C)dt\,, 
   	\end{align*}
   	which is the formulation of \ref{mom} in the sense of distributions.
   	
   	Finally, let us choose a test vector field of the form $(t,x,v)\mapsto\chi(t)\phi(|v|^2)$ where $\chi\in C^\infty_c((0,+\infty))$ (since the left-hand side of the advection equation satisfied 
   	by $f_\eps$ is a compactly support distribution of order $1$ when $t$ is restricted to a bounded interval of $(0,+\infty)$, one can use a test function whose dependence in $v$ is of
   	class $C^1$ only: see Remark 8.1.1 in \cite{Ambrosio}). Arguing as above leads to the equality
   	\begin{align*}
   		& \int_{(0,+\infty)\times\R^d\times\R^d} \chi'(t)\phi(|v|^2)f_\eps(t,dxdv;C)dt
   		\\
   		& =2\int_{(0,+\infty)\times\R^d\times\R^d}\chi(t)\kappa_\eps(x,C)S(|v|)|v|^2\phi'(|v|^2)f_\eps(t,dxdv;C)dt\,.
   	\end{align*}
   	Hence
   	\begin{align*}
   		\frac{d}{dt} & \int_{\R^d\times\R^d}\phi(|v|^2)f_\eps(t,dxdv;C)
   		\\
   		& =-2\int_{\R^d\times\R^d}\kappa_\eps(x,C)S(|v|)|v|^2\phi'(|v|^2)f_\eps(t,dxdv;C)dt\,,
   	\end{align*}
   	in the sense of distributions on $(0,+\infty)$. Since $\phi$ is nondecreasing and of class $C^1$ on $[0,+\infty)$, its derivative $\phi'$ is nonnegative, which implies \ref{funz}.
   \end{proof}
   

\begin{remark}
	\label{suppcomp}
	Since, as a consequence of Lemma \ref{cor}, $\Supp(\E[f_\eps(t,\cdot;C)])\subset\overline{B(0,R+Rt)}\times\overline{B(0,R)}$, whenever $f_\eps^0(\cdot;C)\in \cM^1_+(\R^d\times\R^d)$ each vague limit point of the Radon measure  $\E[f_\eps(t,\cdot;C)]$ on $\R^d\times\R^d$ is also a weak limit point. The same remark is true for the Radon measure  $\E[f_\eps(t,\dot;C)]dt$ on $[0,T]\times \R^d\times\R^d$.
\end{remark}

We are now ready to establish some properties of the limit points of 
$\E[f_\eps(t,\cdot;C)]$ in the weak-* (and therefore, because of the previous Remark \ref{suppcomp}, weak) topology.

\begin{prop}
	\label{conslimp}
	Assume that $f_\eps^0(\cdot;C)$ is supported in $\overline{B(0,R)}\times\overline{B(0,R)}$ for some $R>0$ and all configuration $C$ of inclusions, and that 
	\[
	\sup_{\eps>0}\E\int_{\R^d\times\R^d}f_\eps^0(dxdv;C)<+\infty\,.
	\]
	Then the family of positive measures $\E[f_\eps(t,\cdot;C)]$ on $\R^d\times\R^d$ is weakly relatively compact for each $t\ge 0$ and $\E[f_\eps]$ 
	is weakly relatively compact in the set of positive measures with finite total mass on $[0,T]\times\R^d\times\R^d$ for each $T>0$. Any of its limit points $F$
	as $\eps\to 0^+$ satisfies
	\be 
	\label{limitec}
	\partial_t\int_{\R^d}F(t,\cdot)dv+\nabla_x\cdot\int_{\R^d} F(t,\cdot)v dv=0
	\ee
	in $\mathcal D'((0,+\infty)\times\R^d)$, and
	\be
	\label{limitem}
	\frac{d}{dt}\int_{\R^d\times\R^d}F(t,dxdv)=0
	\ee
	in $\mathcal D'((0,+\infty))$.
	
\end{prop}

\begin{proof}
	
	By construction, $t\mapsto\E[f_\eps(t,\cdot;C)]$ is a family of continuous functions on $[0,+\infty)$ in the set of positive measures on $\R^d\times\R^d$ 
	equipped with the weak topology, and Proposition \ref{apriori1} \ref{mass} shows that
	\[
	\sup_{\eps>0}\E\int_{\R^d\times\R^d}f_\eps(t,dxdv;C)=\sup_{\eps>0}\E\int_{\R^d\times\R^d}f_\eps^0(dxdv;C)<+\infty\,.
	\]
	The weak-* relative compactness of $\E[f_\eps]$ on $[0,T]\times\R^d\times\R^d$ and  of $\E[f_\eps(t,\cdot;C)]$ on $\R^d\times\R^d$ 
	for each $t\ge 0$ follow from statement (13.4.2) in \cite{Dieudonne2}. Since the weak-* closures of both these families are metrizable, any element of these closures is a weak-*
	limit of a subsequence of $\E[f_\eps(t,\cdot;C)]$ for fixed $t>0$, or of $\E[f_\eps]$ as $\eps\to 0$. Finally, Lemma \ref{cor} implies that
	\[
	\mathrm{supp}(\E[f_\eps(t,\cdot;C)]\subset\overline{B(0,R+Rt)}\times\overline{B(0,R)}
	\]
	for all $\eps>0$ and all $t\in[0,T]$, so that weakly-* converging subsequences of $\E[f_\eps(t,\cdot;C)]$ for fixed $t>0$, or of $\E[f_\eps]$ as $\eps\to 0$ are
	in fact weakly convergent.
	
	Thus, if $F$ is a limit point of $\E[f_\eps]$, there exists a subsequence $\E[f_{\eps_n}]\to F$ weakly in $[0,T]\times\R^d\times\R^d$, and in particular
	\[
	\int_{\R^d\times\R^d}\E[f_{\eps_n}](t,dxdv)\to\int_{\R^d\times\R^d}F(t,dxdv)
	\]
	weakly on $[0,T]$ as $\eps_n\to 0$, so that
	\[
	\frac{d}{dt}\int_{\R^d\times\R^d}F(t,dxdv)=0
	\]
	by passing to the limit in Proposition \ref{apriori1} \ref{mass} in the sense of distributions on $(0,+\infty)$. Similarly $|v|\E[f_{\eps_n}]\to |v|F$ weakly in $[0,T]\times\R^d\times\R^d$, 
	and in particular
	\[
	\int_{\R^d}(1,v)\E[f_{\eps_n}](t,\cdot)dv\to\int_{\R^d}(1,v)F(t,\cdot)dv
	\]
	weakly on $[0,T]\times\R^d$ as $\eps_n\to 0$. Hence
	\[
	\partial_t\int_{\R^d}F(t,\cdot) dv+\nabla_x\cdot\int_{\R^d}F(t,\cdot)v dv=0
	\]
	by passing to the limit in the sense of distributions in $(0,+\infty)\times\R^d$.
	\end{proof}

The last fundamental properties of the limit points of  $\E[f_\eps(t,\cdot;C)]$ concern directly the structure of the limit equation and the uniqueness of its solution.

For the sake of simplicity and clarity we shall choose, from now on, $f_\eps(0,\cdot;C)=f_0(\cdot)\in \cM^1_+(\R^d\times\R^d)$: the results are valid also in the generic case, with suitable assumptions on the sequence $f_\eps(0,\cdot;C)$.

	\begin{lem}
	\label{L-MassCons2}

Assume that $f_\eps^0=f^0$ is a positive measure of finite mass supported in $\overline{B(0,R)}\times\overline{B(0,R)}$ for some $R>0$, independent of the configuration 
$C$ of inclusions, and let $F$ be a weak limit point of $\E[f_\eps]$ as $\eps\to 0$. Assume the existence of $\ell\in C(\R^d)$ such that $\ell\ge 0$ on $\R^d$, and of
a time-dependent positive Radon measure $m_F(t,\cdot)$ on $\R^d\times\R^d$ supported in $\R^d\times\overline{B(0,R)}$ for all $t\ge 0$ such that
\be
\label{equalimn0}
(\partial_t+v\cdot\nabla_x+\ell(v))F=m_F\quad\text{ in }\mathcal D'((0,+\infty)\times\R^d\times(\R^d\setminus\{0\}))\,.
\ee
Then
\[
(\partial_t+v\cdot\nabla_x+\ell(v))F-m_F=\lambda_F(t,dx)\otimes\delta_0(dv)\quad\text{ in }\mathcal D'((0,+\infty)\times\R^d\times\R^d)\,,
\]
where $\lambda_F(t,\cdot)$ is a time-dependent Radon measure on $\R^d$.
\end{lem}

\begin{proof}
	
	Consider the distribution 
	\[
	S=(\partial_t+v\cdot\nabla_x+\ell(v))F-m_F\in\mathcal D'((0,+\infty)\times\R^d\times\R^d)\,.
	\]
	It is assumed in the statement of the lemma that 
	\[
	S\big|_{(0,+\infty)\times\R^d\times(\R^d\setminus\{0\})}=0\,.
	\]
	Hence $\mathrm{supp}(S)\subset(0,+\infty)\times\R^d\times\{0\}$. By Theorems XXXV-XXXVI in chapter III of \cite{Schwartz},
	\[
	S=(\partial_t+v\cdot\nabla_x+\ell(v))F-m_F=\sum_{\text{locally finite}}A_\alpha\otimes\partial_v^\alpha\delta_0(dv)
	\]
	where $A_\alpha\in\mathcal D'((0,+\infty)\times\R^d)$. 
	
	Consider next, for each $\phi\in C^\infty_c((0,+\infty)\times\R^d)$, the linear functionals
	\[
	S_1:\,\psi\mapsto-\int_{(0,+\infty)\times\R^d\times\R^d}(\partial_t+v\cdot\nabla_x)\phi(t,x)\psi(v)F(t,dxdv)dt
	\]
	and
	\[
	S_2:\,\psi\mapsto\int_{(0,+\infty)\times\R^d\times\R^d}\phi(t,x)\psi(v)(\ell(v)F(t,dxdv)-m_F(t,dxdv))dt\,.
	\]
	These linear functionals are Radon measures on $\R^d$ with finite total variation supported in $\overline{B(0,R)}$. Therefore
	\[
	\langle S,\phi\otimes\psi\rangle=\langle S_1,\psi\rangle+\langle S_2,\psi\rangle\,,
	\]
	and since $S_1$ and $S_2$ are distributions of order $0$,
	\[
	|\alpha|>0\implies A_\alpha=0\,,
	\]
	and hence
	\[
	S=(\partial_t+v\cdot\nabla_x+\ell(v))F-m_F=A_0\otimes\delta_0\quad\text{ in }\mathcal D'((0,+\infty)\times\R^d\times\R^d)\,.
	\]
	By Proposition \ref{apriori1} \ref{cont}, one has $\langle S_1,1\rangle=0$, so that
	\[
	\begin{aligned}
		\langle S,\phi\otimes 1\rangle=&\langle S_2,1\rangle
		\\
		=&\!\!\!\int_{(0,+\infty)\times\R^d\times\R^d}\!\!\!\!\!\!\!\!\!\phi(t,x)(\ell(v)F(t,dxdv)\!-\!m_F(t,dxdv))dt=\!\!\langle A_0,\phi\rangle\,,
	\end{aligned}
	\]
	meaning that
	\[
	\lambda_F(t,\cdot)=A_0(t,\cdot)=\int_{\R^d}\ell(v)F(t,\cdot)dv-\int_{\R^d}m_F(t,\cdot)dv\,.
	\]
	In particular $\lambda_F$ is a time-dependent Radon measure.
		
\end{proof}

\begin{remark}
	Notice that the fact that the distribution $\l_F$ is of order $0$ is a consequence of the mass conservation property valid for the limit point $F$ and proved in Proposition \ref{conslimp}.
\end{remark}


\smallskip
Let $\ell\in C_b(\R^d)$ and let $\mathcal K:\,C_0(\R^d)\to C_0(\R^d)$ be a continuous linear map. The notation $C_0(\R^d)$ designates
the Banach space of real-valued continuous functions on $\R^d$ vanishing at infinity, equipped with the norm defining uniform convergence on $\R^d$. Thus 
the adjoint of $\mathcal K$ is a continuous linear map $\mathcal K^*:\,\mathcal M^1(\R^d)\to\mathcal M^1(\R^d)$, since the topological dual of 
$C_0(\R^d)$ is $\mathcal M^1(\R^d)$ by Theorem 6.6 in chapter II of \cite{Malliavin}. Henceforth we denote by $\tilde{\mathcal K}$ the operator defined 
on $C_0(\R^d\times\R^d)$ by the formula
\[
(\tilde{\mathcal K}\psi)(x,v):=(\mathcal K\psi(x,\cdot))(v)\,.
\]

\noindent

\begin{lem}
	\label{unic}
	Let $F(t,\cdot)$ be a time-dependent Radon measure on $\R^d\times\R^d$ with finite total variation, continuous on finite time intervals for the weak 
	topology of $\mathcal M^1(\R^d\times\R^d)$ and such that $F(t,\cdot)$ has compact support for all $t\ge 0$. Assume that $F$ is a weak solution to
	\[
	\left\{
	\begin{aligned}
		{}&(\partial_t+v\cdot\nabla_x)F+\ell F=\tilde{\mathcal K}^*(\ell F)\,,\qquad x,v\in\R^d\,,\,\,t>0\,,
		\\
		&F(0,\cdot)=0\,.
	\end{aligned}
	\right.
	\]
	Then $F(t,\cdot)=0$ for all $t\ge 0$.
\end{lem}

\begin{proof}
	
	Let $T>0$, and let $\phi\in C_0([0,T]\times\R^d\times\R^d)$ such that $\partial_t\phi$ and $\partial_{x_j}\phi$ exist and belong to $C([0,T]\times\R^d\times\R^d)$ 
	for $j=1,\ldots,d$. Using $\phi$ as a test function in the weak formulation of the Cauchy problem verified by $F$ shows that
	\begin{align*}
		\int_{\R^d\times\R^d}\!\!\! \phi(&T,x,v) F(T,dxdv)-\!\!\!\int_0^T\!\!\!\int_{\R^d\times\R^d}\!\!\!((\partial_t+v\cdot\nabla_x)\phi(t,x,v)F(t,dxdv)dt
		\\
	\!\!\!\!	&=\int_0^T\!\!\!\int_{\R^d\times\R^d}\ell(v)(\mathcal K(\phi(t,x,\cdot))(v)-\phi(t,x,v))F(t,dxdv)dt.
	\end{align*}
	As mentioned above, see Remark 8.1.1 in \cite{Ambrosio} for using 
	a class of test functions wider than $C_c^\infty((0,T)\times\R^d)$.
	Besides, since it is known a priori that $F(t,\cdot)$ 
	is weakly continuous in $t\in[0,T]$ with values in the space of compactly supported Radon measures, one uses $\phi(T,\cdot,\cdot)$ as a test function to identify $F(T,\cdot)$, which is the reason 
	for dropping the compact support assumption on $\phi$ in all its variables).
	
	This identity shows that
	\[
	\int_{\R^d\times\R^d}\phi(T,x,v)F(T,dxdv)=0
	\]
	for all $\phi\in C_0([0,T]\times\R^d\times\R^d)$ such that $\partial_t\phi$ and $\partial_{x_j}\phi$ exist and belong to $C([0,T]\times\R^d\times\R^d)$ for all
	$j=1,\ldots,d$, satisfying
	\[
	-(\partial_t+v\cdot\nabla_x)\phi+\ell(v)\phi=\ell(v)\tilde{\mathcal K}\phi\,,\qquad x,v\in\R^d\,,\,\,t\in(0,T)\,.
	\]
	The next lemma shows that $\phi(T,x,v)$ can be chosen to be any function in the space $C_0(\R^d\times\R^d)$ with $\partial_{x_j}\phi\in C_0(\R^d\times\R^d)$ 
	for all $j=1,\ldots,d$. Hence $F(T,\cdot)=0$, and since $T>0$ is arbitrary, this implies that $F(t,\cdot)=0$ for all $t\geq 0$.	
\end{proof}

\noindent
\begin{lem}
	\label{lemunic}
	 Let $T>0$, let $\phi^T\in C_0(\R^d\times\R^d)$, such that $\partial_{x_j}\phi^T$ exist and belong to $C_0(\R^d\times\R^d)$ for $j=1,\ldots,d$.
There exists $\phi\in C_0([0,T]\times\R^d\times\R^d)$ such that $\partial_{x_j}\phi$ exist and belong to $C_0([0,T]\times\R^d\times\R^d)$ for $j=1,\ldots,d$, 
and such that $\partial_t\phi\in C([0,T]\times\R^d\times\R^d)$, solution to the (backward) Cauchy problem
\[
\left\{
\begin{aligned}
	{}&-(\partial_t+v\cdot\nabla_x)\phi+\ell(v)\phi=\ell(v)\tilde{\mathcal K}\phi\,,\qquad x,v\in\mathrm R^d\,,\,\,0<t<T\,,
	\\
	&\phi(T,\cdot)=\phi^T\,.
\end{aligned}
\right.
\]
\end{lem}

\begin{proof}
	A mild solution to the backward Cauchy problem in Lemma \ref{lemunic} is a function $\phi\in C_0([0,T]\times\R^d\times\R^d)$ satisfying the integral equation 
	\begin{align*}
	\phi(t,x,v) = & e^{-\ell(v)(T-t)}\phi^T(x+(T-t)v,v)
	    \\
	    & - \int_t^T\!\!e^{-\ell(v)(s-t)}\ell(v)\tilde{\mathcal K}\phi(s,x+(s-t)v,v)ds\,,
	\end{align*}
	for all $x,v\in\R^d$ and $t\in[0,T]$. 
	
	The family of linear maps defined by the formula
	\[
	S_0(t)\psi(x,v)=e^{-t\ell(v)}\psi(x-tv,v)
	\]
	is a strongly continuous semigroup of bounded operators on $C_0([0,T]\times\R^d\times\R^d)$ with norm $\|S_0(t)\|\le\exp(T\sup_{v\in\R^d}\ell(v))$. By construction
	$\tilde{\mathcal K}$ is a bounded operator on $C_0(\R^d\times\R^d)$ with the same norm as that of the bounded operator $\mathcal K$ defined on $C_0(\R^d)$.
	The existence (and uniqueness) of the mild solution $\phi$ follows from Proposition 1.2 in chapter 3 of \cite{Pazy}. 
	
	The regularity in $x$ follows from applying the same argument to each partial derivative $\partial_{x_j}\phi$ for $j=1,\ldots,d$. 
	
	Finally, the regularity in $t$ follows from deriving in $t$ both sides of the integral equation above:
	\begin{align*}
		\partial_t\phi(t,\!x,\!v) = & \; e^{-\ell(v)(T\!-\!t)}\ell(v)\phi^T\!(x\!+\!(T\!-\!t)v,\!v)
		\\
		& -\!e^{-\ell(v)(T\!-\!t)}v\!\cdot\!\nabla_x\phi^T\!(x\!+\!(T\!-\!t)v,\!v) +\ell(v)\tilde{\mathcal K}\phi(t,x,v)
		\\
		&-\int_t^Te^{-\ell(v)(s-t)}(\ell(v))^2\tilde{\mathcal K}\phi(s,x\!+\!(s\!-\!t)v,v)ds
		\\
		&+\ell(v)\sum_{j=1}^d\int_t^Te^{\ell(v)(s-t)}v_j\tilde{\mathcal K}\partial_{x_j}\phi(s,x\!+\!(s\!-\!t)v,v)ds\,.
	\end{align*}
	Notice that the 2nd and 5th terms on the right-hand side of this identity involve products by the components of $v$, so that $\partial_t\phi$ is only continuous, but  may fail to be bounded,
	let alone converge to $0$ as $|v|\to+\infty$.
\end{proof}

\smallskip
It remains to check that Lemmas \ref{unic} and \ref{lemunic} apply to the equation satisfied by the limit points of $\E[f_\eps]$ as $\eps\to 0$. We shall see that, in that equation, $\ell(v)=\sigma|v|$, where $\sigma$ is a positive constant, so that $\ell$ is obviously
continuous on $\R^d$, but not bounded. However, since it is known a priori that $F(t,\cdot)$ has compact support included in $\overline{B(0,R(1+t))}\times\overline{B(0,R)}$ for 
each $t\ge 0$, one can replace $\ell(v)=\sigma |v|$ with $\ell(v):=\sigma \min(|v|,R+1)$, which is both continuous and bounded on $\R^d$, without affecting $F$.

As for the operator $\mathcal K$, it will turn out to be of the form
\[
\tfrac1{B^{d-1}}\int_{\nu\in B^{d-1}(0,1)}\psi\left(\mathfrak V(|\nu|,v)\right)d\nu+k(|v|)\psi(0)\,,
\]
where
\[
\mathfrak V(\beta,v):=\mathcal V(\beta,|v|)\tfrac{v}{|v|}\,.
\]
In the definition of $\mathfrak V$, the function $\mathcal V:\,[0,1]\times[0,+\infty)\to[0,+\infty)$ is such that 
\[
\left\{
\begin{aligned}
	{}&\mathcal V(\beta,\cdot):\,(c(\beta),+\infty)\to(0,+\infty)\text{ is a }C^1\text{-diffeomorphism, and}
	\\
	&\mathcal V(\beta,\cdot)([0,c(\beta)])=\{0\}\,,\quad\text{ while }\quad\mathcal V(\beta,r)\le r\text{ for all }r\ge 0\,,
\end{aligned}
\right.
\]
with $c\in C^1((0,1))\cap C([0,1])$ nonincreasing. Finally
\[
k(u):=|\{\nu\in B^{d-1}(0,1)\text{ s.t. }|c(\nu)|\ge u\}|/B^{d-1},.
\]
One easily checks that $\mathcal K:\,C_0(\R^d)\to C_0(\R^d)$ is a continuous linear map satisfying
\[
\psi\in C_0(\R^d)\text{ and }\psi\ge 0\text{ on }\R^d\implies{\mathcal K}\psi\ge 0\,,
\]
and
\[
|{\mathcal K}\psi(v)|\le\tilde{\mathcal K}|\psi|(v)\le\sup_{v\in\R^d}|\psi(v)|\,.
\]
Finally, the adjoint of $\mathcal K$ is the continuous linear map $\mathcal K^*:\,\mathcal M^1(\R^d)\to \mathcal M^1(\R^d)$ given by the formula
\[
\mathcal K^*\mu(dw)=\left(\int_{\R^d}k(|v|)\mu(dv)\right)\delta_0(dw)+\tfrac1{B^{d-1}}\int_{|\nu|\le 1}\mathfrak V(|\nu|,\cdot)\#\left(\mu\big|_{\R^d\setminus\{0\}}\right).
\]
In the latter formula, the notation $\Phi\#m$ designates the push-forward of the measure $m$ by the (measurable) map $\Phi$, i.e. $\Phi\#m(B)=m(\Phi^{-1}(B))$ (see for instance
formula 1.4.3 in chapter IV of \cite{Malliavin}), while the notation $m\big|_A$ designates the restriction of the measure $m$ to the measurable set $A\subset\R^d$.

\section{The limit equation}
\label{Equalim}

 In this section we shall derive the limit equation for $\E[f_\eps(t,\cdot;C)]$ in $(0,+\infty)\times\R^d\times\R^d\setminus\{0\}$, i.e. an equation of the form (\ref{equalimn0}), with the help of the technique adopted for instance  in \cite{Gallavotti}, which is based on the construction of a semiexplicit solution to (\ref{evomisura}) making use of the backward flow associated to (\ref{moto}). We shall then apply Lemma \ref{L-MassCons2} to recover the limit equation in $(0,+\infty)\times\R^d\times\R^d$.
 
 Again for the sake of simplicity, we shall assume $$f_\eps(0,dxdv;C)=f_0(x,v)dxdv.$$ Although in the physical model  the light particles should  have $x\in Z_\eps[C]$, this requirement becomes unnecessary when $\eps\to 0^+$, since $P(x\in B(c,\eps))<\eps$; in any case,
 all results can be generalised to a generic sequence of initial data $f_\eps(0,\cdot;C)\in \cM^1_+(\R^d\times\R^d)$, with suitable straightforward assumptions on this sequence.
 
 We define, from (\ref{defevomis}), the evolution of 
 $$
 F_\eps(t,dxdv)=\E[f_\eps(t,dxdv;C)] ,
 $$
and analyse  its limit evolution equation in $(0,+\infty)\times\R^d\times\R^d\setminus \{0\}$, in order to apply Lemma \ref{L-MassCons2}.
 
 We study then
 \be
 \label{defevomise}
 \int F_\eps(t,dxdv)\phi(x,v):=\int dx dv f_0(x,v)\E[\phi_\eps(t,x,v;C)]
 \ee
for $\phi\in C_c(\R^d\times\R^d\setminus \{0\})$ (where $F_\eps(0,dxdv)=f_0(x,v)dxdv$).
 
 Call $(X^\eps_{-t},V^\eps_{-t})(x,v;C)$ the inverse flow associated to the equations of motion (\ref{moto}).
 Then,    
 \begin{align}
 \label{evocompl}
 \int &F_\eps(t,dxdv)\phi(x,v)=\int dx dv F_\eps(t,x,v)\phi(x,v) 
       \\
 & =  \int dxdv \, \E[f_0((X_{-t}^\eps, V_{-t}^\eps)(x,v;C) ) J_{\eps,C}(t,x,v)\indc_{|v|>0}]\phi(x,v) 
 \nonumber
\end{align}where $J_{\eps,C}(t,x,v)$ is the Jacobian of the change of coordinates $(x,v)\to (X_{t}^\eps, V_{t}^\eps)(x,v;C)$
 (we shall see in Subsection \ref{Jacsec} that the inverse flow is defined whenever $|v|>0$). 

An equivalent expression for $F_\eps(t,x,v)$ may be derived from (\ref{evomisura}), so that we may write 
\begin{align}
{F}_\eps 
    &  (t,x,v) \nonumber \\
    & = \E [f_0(\!(X_{-t}^\eps, V_{-t}^\eps)(x,v;C)\!) J_{\eps,C}(t,x,v)\indc_{|v|>0} ]
            \label{va1}  \\
    & = \E  \bigg[ f_0 \Big(\!(X_{-t}^\eps,\! V_{-t}^\eps)(x,\!v;C) \Big)
             \label{va2}             \\
        & \times \!\! \indc_{|v|>0} e^{d \!\int_{0}^{t}ds\ka_\eps( \! X_{-s}^\eps \!(x,v;C);C) \!
        \big[ \! S(| V_{-s}^\eps(x,v;C)|)+ V_{-s}^\eps(x,v;C)\cdot\grad_V S(| V_{-s}^\eps(x,v;C)|)\big]} \!\bigg]\!.
       \nonumber
\end{align}
In next subsections we shall define all the elements which will allow us to make explicit (\ref{va1}).

\subsection{The backward flow}

Since ${\bf P}(x\in B_\eps(c))<\eps$, without loss of generality, we shall assume $x\in Z_\eps[C]$.

Define
\begin{align*}
\Big\{ t\in[0,+\infty) &  \; \text{ s.t. } \; X_{t}^\eps (x,v;C)\notin Z_\eps[C] \Big\}
   \\[1mm]
 & = [\tau^1_-(x,v),\tau^1_+(x,v)]\cup[\tau^2_-(x,v),\tau^2_+(x,v)]\cup\ldots
\end{align*}
with
$$
0<\tau^1_-(x,v)<\tau^1_+(x,v)<\tau^2_-(x,v)<\tau^2_+(x,v)<\ldots\,.
$$
and the assumption that 
$\tau^k_+(x,v)=+\infty$ implies that $[\tau^l_-(x,v),\tau^l_+(x,v)]=\varnothing$ for all $l>k$ (here  $\tau_\pm^j(x,v)=\tau_\pm^j(x,v;C)$, but we drop the dependence on $C$ in order to simplify the notation).

The times $\tau^j_-(x,v)$ and $\tau^j_+(x,v)$ are respectively the entering and the exit time in the $j$-th connected region occupied by one or more obstacles along the forward flow (\ref{moto}). 

The number of connected regions occupied by the obstacles crossed by $\{X_{s}^\eps (x,v;C)\}_{s\in(0,t)}$ is therefore given by 
$$N(t,x,v;C) =\max_l \{l \quad \mathrm{s.t.}\quad\tau^l_-(x,v)<t\}.$$

 Consider $(x,v)=(X_{t}^\eps, V_{t}^\eps)(y,z;C)$ for $(y,z)\in \R^d\times\R^d$ and denote $N=N(t,y,z;C)$.

We may define 
\begin{eqnarray*}
	\tau^l_{i,+}(x,v)&=&t-\tau^{N+1-l}_+(y,z)\\
	\tau^l_{i,-}(x,v)&=&t-\tau^{N+1-l}_-(y,z)\\
	\tau^0_{i,-}(x,v)&=&0\\
	\tau^{N+1}_{i,+}(x,v)&=&t
\end{eqnarray*}
with
$$
0<\tau^1_{i,+}(x,v)<\tau^1_{i,-}(x,v)<\tau^2_{i,+}(x,v)<\tau^2_{i,-}(x,v)<\ldots\,.
$$

The backward flow is then obtained as the solution $(Y_{t}^\eps, Z_{t}^\eps)(x,-v;C)$ to the equations of motion 
\be
	\label{motoinv}
	\left\{
	\begin{array}{l}
		\dot Y=Z\\
		\\
		\dot Z=\frac{\ka}\eps\indc_{Z^c_\eps[C]}(Y)S(|Z|)Z=\ka_\eps(Y;C)S(|Z|)Z\\
		\\
		(Y,Z)\big|_{t=0}=(x,-v),
	\end{array}
	\right.
\ee
i.e. $X_{-t}^\eps(x,v;C)=Y_{t}^\eps(x,-v;C)$, while $V_{-t}^\eps(x,v;C)=-Z_{t}^\eps(x,-v;C)$.

In particular, we may express the solution to (\ref{motoinv}) in semiexplicit form as
\begin{eqnarray}
Y_{t}^\eps (x,-v;C) \!\!\!&\!\!=\!\!&\!\!\!  X_{-t}^\eps (x,v;C) 
        \nonumber. \\[2mm] 
        \!\!&\!=\!&\!\!    x + \! \sum_{j=0}^{N-1} 
        Z(\tau^{j+1}_{i,+})(\tau^{j+1}_{i,+} - \tau^{j}_{i,-}) +
        Z(\tau^{N+1}_{i,+})(t - \tau^{N}_{i,-})
        \nonumber \\
        & & -\sum_{j=1}^{N}|Y_{\tau^{j}_{i,-}}^\eps (x,-v;C)-Y_{\tau^{j}_{i,+}}^\eps (x,-v;C)|\hat{v}
        \label{flussoind} \\[3mm]
Z_{t}^\eps(x,-v;C)  \!\!&\!=\!&\!\!  -V_{-t}^\eps(x,v;C)
        \nonumber \\[2mm]
         \!\!&\!=\!&\!\!  -\left(v+\hat{v}\sum_{j=1}^{N}
            \frac{\ka}{\eps}\int_{\tau^{j}_{i,+}}^{\tau^{j}_{i,-}}S(|Z|(s))|Z(s)|ds\right) ,
            \nonumber
\end{eqnarray}
where $Z(t)=Z_{t}^\eps(x,-v;C)$ and $\hat{v}=\frac{v}{|v|}$.

For ease of notation, from now on, we shall drop the dependence on the configuration of obstacles $C$ in the flows. 

\subsection{The Jacobian $J_{\eps,C}(t,x,v)$}
\label{Jacsec}

We may now compute the Jacobian $J_{\eps,C}(t,x,v)$, observing that, from (\ref{va1})-(\ref{va2}), we deduce that:
\begin{align}
J_{\eps,C} & (t,x,v) \nonumber \\
     & =e^{d \int_{0}^{t}ds\ka_\eps(X_{-s}^\eps (x,v);C)
     \big[S(| V_{-s}^\eps(x,v)|)+ V_{-s}^\eps(x,v)\cdot\grad_V S(| V_{-s}^\eps(x,v)|)\big]}.
     \label{jac0}
\end{align}
This expression can be made more explicit since, when \newline $Y_{s}^\eps (x,-v)\in Z^c_\eps[C]$,
$$
\frac{\ka}{\eps} S(|Z|)=\frac{\dot{|Z|}}{|Z|},
$$
so that 
\begin{align*}
\exp\left(d \int_{0}^{t}\ka_\eps(Y_{s}^\eps (x,-v);C)S(| Z_{s}^\eps(x,-v)|)\right) 
       & =\left(\frac{|Z^\eps_{\tau^N_{i,-}}|}{|v|}\right)^{\!\!d} \\
       & =\left(\frac{| Z_{t}^\eps(x,-v)|}{|v|}\right)^{\!\!d}
\end{align*}
and
$$
\exp\!\left( \int_{0}^{t} \!\! ds \, \ka_\eps(X_{-s}^\eps (x,v);C)V_{-s}^\eps(x,v)\cdot\grad_V S(| V_{-s}^\eps(x,v)|)\right)
       \!\!=\!\! \frac{S(| Z_{t}^\eps(x,-v)|)}{S(|v|)}.
$$
We have then 
\be
\label{jac1}
J_{\eps,C}(t,x,v)=\left(\frac{| Z_{t}^\eps(x,-v)|}{|v|}\right)^d\frac{S(| Z_{t}^\eps(x,-v)|)}{S(|v|)},
\ee
while the Jacobian of the change of coordinates $(x,v)\to (Y_{t}, -Z_{t})$ is $\frac 1{J_{\eps,C}}$.

The change of coordinates $(x,v)\to X_{t}^\eps (x,v), V_{t}^\eps(x,v)$ is therefore well defined when $|v|>0$.

\subsection{The flow tube}
In order to perform explicit computations in the expectations, it is  useful to define the notion of \textit{flow tube} associated to a space trajectory .

On specific configurations of obstacles, this notion allows to simplify remarkably  the expression of $\E$, giving also an intuition on the behaviour of the particle system.

We define here the flow tube associated to $(Y_{t}^\eps (x,-v), Z_{t}^\eps(x,-v))$ as 
\be
\label{tubo}
\mathcal{T}_\eps (-t;x,v)(C):= \{y\in\R^d : |y-Y_{s}^\eps (x,-v)|<\eps \,\,\,s\in[0,t]\}.
\ee
A given trajectory $Y_{t}^\eps (x,-v)$ will meet an obstacle $c$ only if $c\in \mathcal{T}_\eps (-t;x,v)(C)$.

Whenever $| V_{-t}^\eps(x,v) |=| Z_{t}^\eps(x,-v) |< V_{\max}$, we have the following inclusion
\be
\label{tubo2}
	\mathcal{T}_\eps (-t;x,v)(C)\subset
	 \mathcal{T}_\eps (-t;x,\hat{v})(V_{\max}) 
\ee
with 
\begin{align*}
 \mathcal{T}_\eps (-t;x,\hat{v})(V_{\max}) 
         & :=\Big\{y\in\R^d : |y-(x-   V_{\max} \hat{v} s)|<\eps \,\,\,s\in[0,t]\Big\}  \\[1mm]
	 & \subset B(x,V_{\max}t).
\end{align*}
In the same way we may define the flow tube associated to the direct flow to $(X_{t}^\eps (x,v), V_{t}^\eps(x,v))$ as 
\be
\label{tubod}
\mathcal{T}_\eps (t;x,v)(C):= \{y\in\R^d : |y-X_{s}^\eps (x,v)|<\eps \,\,\,s\in[0,t]\}.
\ee

We have then:
\be
	\label{tubodir}
\begin{array}{c}
	\mathcal{T}_\eps (t;Y_{t}^\eps (x,-v),-Z_{t}^\eps (x,-v))(C)=
	\mathcal{T}_\eps (t;Y,-Z)(C) =\\
	\\
	\{y\in\R^d : |y-X_{s}^\eps (Y,-Z)|<\eps \,\,\,s\in[0,t]\}, 
\end{array}
\ee

and obviously, whenever $| V_{-t}^\eps(x,v) |=| Z_{t}^\eps(x,-v) |< V_{\max}$,
\[
\label{tubiug}
\mathcal{T}_\eps (t;Y,-Z)(C)=\mathcal{T}_\eps (-t;x,v)(C) \subset B(x,V_{\max}t)
\]
and 
\be
	\label{tubodir1}
	\mathcal{T}_\eps (t;Y,-Z)(C)\subset \mathcal{T}_\eps (t;Y,\hat{v})(V_{\max}) 
	=\mathcal{T}_\eps (-t;Y,-\hat{v})(V_{\max})
	\ee
	\[\subset B(x,V_{\max}t).
\]

These inclusions will be useful to simplify computations.

\subsection{The Markovian set $A^\eps_1(x,v)$ and the limit dynamics}
\label{markov}

For each $(x,v)\in\R^d\times\R^d$ and $t\in (0,+\infty)$, call
\be
\label{A1}
A^\eps_1(x,v)= \{C | \forall c,c'\in C s.t. (x-(0,+\infty) v)\cap B(c,\eps)\not=\varnothing \,\, \mathrm{and} 
\ee
$$
 (x-(0,+\infty) v)\cap B(c',\eps)\not=\varnothing \implies B(c,\eps)\cap B(c',\eps)=\varnothing\}.
$$
Such configurations of obstacles are the ones where the obstacles crossed by any trajectory with initial datum $(x,v)$, defined by (\ref{motoinv}), do not superimpose.

The set $A^\eps_1(x,v)$ enjoys some relevant properties for the particle system with dynamics (\ref{moto}).

\subsubsection{Irrelevance of $(A^\eps_1(x,v))^c$ when $\eps\to 0^+$}
\label{mis0}

The first relevant property of $A^\eps_1(x,v)$  is given by the following lemma.

\begin{lem}
	\label{sovrapp}
	Assume that $f_0(x,v)\in L^1_+(\R^d\times\R^d)$ with $\Supp(f_0)\subset \overline{B(0,R)}\times\overline{B(0,R)}$. 
	Then for each $\phi\in C_b(\R^d\times\R^d)$, $\forall t>0$
	\be
	\label{collmult} 
	\lim_{\eps\to 0^+} 
	\int_{\R^d\times\R^d} dxdv \phi(x,v)
	\E[f_\eps(t,x,v) \indc_{|v|>0}\indc_{(A^\eps_1)^c}]=0
	\ee
\end{lem}
             
\begin{proof}

We have a.e. in $\R^d\times\R^d$
\begin{align*}
\E[f_\eps(t,x,v) & \indc_{|v|>0}  \indc_{(A^\eps_1)^c}] \\
       & \leq  \sum_{M\geq 0}\frac{\l_{\eps}^M}{M!}\int dc_1\ldots dc_M e^{-\l_{\eps}|B(x,Rt)|} \\
       & \times  \prod_{j=1, j\neq 2}^{M}\indc_{B(x,R t)}(c_j) \indc_{B(c_1,2\eps)}(c_2)\indc_{|v|>0} \\[2mm]
       & \times f_0((Y_{t}^\eps (x,-v),-Z_{t}^\eps(x,-v))) 
       J_{\eps,C}(t,x,v).
\end{align*}

Therefore, because of definition (\ref{tubo}) and inclusion (\ref{tubodir1}), 
\begin{align*}
   \int dv dx \, & \phi(x,v)\E[f_\eps(t,x,v)\indc_{|v|>0} \indc_{(A^\eps_1)^c}] \\
   & \leq  \sum_{M\geq 0}\frac{ \l_{\eps}^M}{M!}\int dYdZ 
      \int dc_1\ldots dc_M e^{-\l_{\eps}|\mathcal{T}_\eps (t;Y,\hat{v})(R)|}  \\[1mm]
    & \times \prod_{j=1, j\neq 2}^{M}\indc_{\mathcal{T}_\eps (t;Y,\hat{v})(R)}(c_j) 
        \indc_{B(c_1,2\eps)}(c_2)\indc_{|v|>0}  \\[2mm]
    & \times f_0(Y, Z) \phi\big(X^\eps_t(Y,Z),V^\eps_t(Y,Z)\big),
\end{align*}
where we have performed the change of variable $(x,v)\to (Y,Z)=(Y_{t}^\eps (x,-v), -Z_{t}^\eps(x,-v))$, with Jacobian $J_{\eps,C}^{-1}$.

We have then 
\begin{align*}
	\int dv dx \, & \E[f_\eps(t,x,v)\indc_{|v|>0} \indc_{(A^\eps_1)^c}] \\
	& \leq  \|\phi\|_\infty	\l_\eps |B(0,2\eps)|\int dYdZ f_0(Y, Z)  \\
	& \leq \lambda \eps 2^d B^{d} \|\phi\|_\infty 
             \int dYdZ	f_0(Y, Z) \stackrel{\eps\to0}{\to}0
\end{align*}
and the lemma is proved.
\end{proof}
	
	We have then that $\E[f_\eps(t,x,v) \indc_{|v|>0}\indc_{(A^\eps_1)^c}]\deb 0$, and therefore $A^\eps_1 (x,v)$ is the only set of configurations of obstacles which play a role in the computation of  the weak limit of $\E[f_\eps(t,x,v) \indc_{|v|>0}]$ when $\eps\to 0^+$.
	
	\bigskip
	
	\subsubsection{The dynamics on $A^\eps_1(x,v)$}
	A second property of $A^\eps_1(x,v)$ is that,
	when $C\in A^\eps_1(x,v)$, the motion described by (\ref{motoinv}) is quite simplified.
	
	Within this set, the trajectory starting at $(x,-v)$ will cross  $N=N(t,Y_{t}^\eps (x,-v),-Z_{t}^\eps(x,-v);C)$ single obstacles, which can be ordered according to their collision times; therefore $\tau^j_{i,+}(x,v)$ and $\tau^j_{i,-}(x,v)$ are resp. the entry and the exit time of the trajectory in and from the $j$-th obstacle.
	
	For $C\in A^\eps_1(x,v)$ such that $C\bigcap \mathcal{T}_\eps (-t;x,v)(C)=(c_1, \ldots, c_N)$, where $c_1, \ldots, c_N$  are ordered according to their collision times, we call:

\begin{align}
[C]_N & \!= \! (c_1,\ldots,c_N) \nonumber \\
	  & \!=\! \left\{C' \!\!\in\! A^\eps_1(x,v) \!: \exists i_1,\ldots,i_N \!\in\! \N \,\,
	  s.t.\,\,(c'_{i_1}, \ldots, c'_{i_N}) \!=\! (c_1,\ldots,c_N) \right\} \! .
	  \label{cleq}
\end{align}
	
	We call $$[A^\eps_1(t,x,v)]_N = \left\{[C]_N\, \mathrm{defined}\,\mathrm{by} (\ref{cleq})\right\}$$
	and define the set of ordered obstacles possibly met by the trajectories (\ref{motoinv}) up to time $t$ by
	$$
	[A^\eps_1(t,x,v)]=\bigcup\limits_{N=0}^{\infty}[A^\eps_1(t,x,v)]_N
	$$
	
	Once again, since ${\bf P}(x\in B_\eps(c))<\eps$, without loss of generality, we shall assume $x\in Z_\eps[C]$.
	
	For a given configuration of ordered obstacles $(c_1,\ldots, c_{N})\in [A^\eps_1(t,x,v)]_{N} $ we define $n_j(x,v)=(n_{j,x_1},n_{j,x_2},\ldots,n_{j,x_d})$ as the unit vector such that
	$$
	c_j + \eps n_j= Y_{\tau^j_{i,+}}^\eps (x,-v)
	$$
	Then, choosing $\hat{v}$ as the unit vector of the first coordinate axis, we define the \textit{impact parameter} of the trajectory with the $j$-th obstacle as
	$$
	h_j =\sqrt{\sum_{i=2}^{d} n_{j,x_i}^2} .
	$$	
	The solution to (\ref{motoinv}) on the set of configurations of obstacles
	$$
	[A^\eps_1(t,x,v)]_N\subset [A^\eps_1(t,x,v)] 
	$$
	can be expressed as:
	\begin{eqnarray}
		Y_{t}^\eps (x,-v) \!\!\!&\!=\!&\!\!\! x \!+\! \sum_{j=0}^{N-1}
		Z^\eps_{\tau^{j+1}_{i,+}}(x,-v)(\tau^{j+1}_{i,+}-\tau^{j}_{i,-})+
		Z^\eps_{\tau^{N+1}_{i,+}}(x,-v)(t-\tau^{N}_{i,-})
		\nonumber \\
		&-&2\eps \sum_{j=1}^{N-1}\sqrt{1-h^2_j} \hat{v}
		-|Y_{\tau^{N}_{i,-}}^\eps (x,-v)-Y_{\tau^{N}_{i,+}}^\eps (x,-v)|\hat{v}
		\label{flussoinda1} \\
		Z_{t}^\eps(x,-v) \!\!\!&\!=\!&\!\!\! -\left(v+\hat{v}\sum_{j=1}^{N}
		\frac{\ka}{\eps}\int_{\tau^{j}_{i,+}}^{\tau^{j}_{i,-}}S(|Z|(s))|Z(s)|ds\right)
		\nonumber
	\end{eqnarray}
	with $|Y_{\tau^{N}_{i,-}}^\eps (x,-v)-Y_{\tau^{N}_{i,+}}^\eps (x,-v)|<2\eps. $

	Within $[A^\eps_1(t,x,v)]_N$, calling $n_j=(\sqrt{1-\nu_j^2}, \nu_{j})$ we may perform the change of variable
	\be
	\label{cambvar}
	(c_1\ldots,c_N)\to (\nu_1, \tau^{1}_{i,+},\ldots, \nu_N, \tau^{N}_{i,+}),
	\ee
	where
	$$c_j= x+\sum_{k=0}^{j-1}
	Z^\eps_{\tau^{j+1}_{i,+}}(\tau^{j+1}_{i,+}-\tau^{k}_{i,-})
	-2\eps\sum_{k=1}^{j-1}\sqrt{1-h^2_k}\hat{v}	-\eps n_j .
	$$
	The Jacobian matrix of (\ref{cambvar}) is a block triangular matrix, with
	Jacobian 
	\begin{eqnarray}
		J_{ct}(\nu_1, \tau^{1}_{i,+},\ldots, \nu_N, \tau^{N}_{i,+})= \prod_{j=1}^{N}\eps^{d-1}|Z^\eps_{\tau^{j}_{i,+}}|
	\end{eqnarray}
and each set $[A^\eps_1(t,x,v)]_N$ is expressed, in the new coordinates, as
\\
\begin{align}
[A^\eps_1(t,x,v)]_N = \Big\{ & (\nu_1, \tau^{1}_{i,+},\ldots, \nu_N, \tau^{N}_{i,+})\!  \in \! (\R^{d-1} \! \times \! (0,t])^N\!:
       \nonumber  \\
      & \quad 0<\tau^{1}_{i,+}<\ldots<\tau^{N}_{i,+}\leq t, \;\; |\nu_j|\leq 1,\;\;  j=1,\ldots,N \Big\}.
      \label{parttubo2}
\end{align}	
The size of the flow tube associated to the trajectory (\ref{flussoinda1}) is
\begin{align}
|\mathcal{T}_\eps (-t;x,v)(c_1\ldots c_N)
              & = B^{d-1}\eps^{d-1} \bigg[\sum_{j=0}^{N}
		|Z|^\eps_{\tau^{j+1}_{i,+}}(\tau^{j+1}_{i,+}-\tau^{j}_{i,-})
		\nonumber \\
		& + 2 \eps \! \sum_{j=1}^{N-1} \! \sqrt{1-h^2_j}
		+ \Big| Y_{\tau^{N}_{i,-}}^\eps (x,-v)-Y_{\tau^{N}_{i,+}}^\eps (x,-v) \Big| \bigg] .
		\label{tubo1coll}
\end{align}
	Finally, thanks to (\ref{motoinv}) , we may make explicit, at each complete crossing, the change in velocity, Indeed,
	\be
	\label{velimpl}
	\int_{|Z|^\eps_{\tau^{j}_{i,+}}}^{|Z|^\eps_{\tau^{j}_{i,-}}}\frac{du}{S(u)}=\frac{\ka}{\eps}|Y_{\tau^{j}_{i,-}}^\eps (x,-v)-Y_{\tau^{j}_{i,+}}(x,-v)|=2\ka\sqrt{1-h_j^2}.
	\ee
	
	Defining, for $z\in(0,+\infty)$, the invertible (because $S(u)>0$)
	function
	$$
	a(z)=\int_{0}^{z} du \frac{1}{S(u)},
	$$
	we may write (\ref{velimpl}) as
	$$
	a(|Z|(\tau^{j}_{i,-}))-a(|Z|(\tau^{j}_{i,+}))=2\ka\sqrt{1-h_j^2}
	$$
	and express the exit velocity from the obstacle $j$ as
	\be
	\label{velusc}
	|Z|^\eps_{\tau^{j}_{i,-}}=a^{-1}[a(|Z|^\eps_{\tau^{j}_{i,+}})+2\ka\sqrt{1-h_j^2}].
	\ee
	where, in particular, 
	\be
	\label{velusc1}
	|Z|^\eps_{\tau^{1}_{i,-}}=a^{-1}\left[a(|v|)+2\ka\sqrt{1-h_j^2}\right].
	\ee
	From (\ref{velusc}-\ref{velusc1}) we see that $|Z|^\eps_{\tau^{j}_{i,-}}$ is independent of $\eps$.
	
	\begin{remark}
		The independence of the exit velocity of $\eps$ follows from the choice $\ka_\eps =\frac{\ka}{\eps}$, i.e.  the increment in the velocity is inversely proportional to the size of the obstacles.
	\end{remark} 
	
	\subsubsection{The limit dynamics on $[A^\eps_1(t,x,v)]$}
	\label{sdinlim}
	\smallskip
	
	From the formulas given in the previous subsection, it is easy to compute the limit dynamics on the relevant set of obstacles.
	
	By (\ref{velusc}-\ref{velusc1}) and $S(u)>S_0$ we may estimate $\tau^{j}_{i,-}-\tau^{j}_{i,+}$ as follows
	\be
	\label{tempicoll}
	\tau^{j}_{i,-}-\tau^{j}_{i,+}=\frac{\eps}{\ka}\int_{|Z|(\tau^{j}_{i,+})}^{|Z|(\tau^{j}_{i,-})}\frac{du}{u S(u)}
	\leq \frac{\eps}{\ka S_0}\ln\left[\frac{|Z|^\eps_{\tau^{j}_{i,-}}}{|Z|^\eps_{\tau^{j}_{i,+}}}\right]
	\ee
	so that
	\be
	\lim_{\eps\to 0^+}	(\tau^{j}_{i,-}-\tau^{j}_{i,+})=0.
	\ee
	
	When $\eps\to 0^+$, we obtain therefore the limit dynamics
	\begin{eqnarray}
		\label{dinlim1}
		Y_{t}^\eps(t,-v) &\to&	Y_t(x,-v)=x+\sum_{j=0}^{N}
		Z^-_{t^{j+1}_{i,+}}(t^{j+1}_{i,+}-t^{j}_{i,+}) \\
		\label{dinlim2}
		Z^\eps_t(x,-v) &\to&Z_t(x,-v)=-v+\sum_{j=1}^{N}\left[Z^+ _{t^{j}_{i,+}}-Z^-_{t^{j}_{i,+}}\right],
	\end{eqnarray}
	where 
	$$
	t^{j}_{i,+}=\lim_{\eps\to 0^+}\tau^{j}_{i,+}
	$$
	and 
	(for a.a. configurations of obstacles)
	$$
	Z^+ _{t^{j}_{i,+}}=-a^{-1}\big[a(|Z^-_{t^{j}_{i,+}}|)+2k\sqrt{1-h_j^2})\big] \hat{v}
	$$
(we note that $\displaystyle Z^\pm_{t^{j}_{i,+}} \!\!=\!\! \lim_{t\to (t^{j}_{i,+})^\pm} \! Z^\eps_t$, recalling that 
$\displaystyle \lim_{\eps\to 0^+}Z^\eps_{\tau^{j}_{i,\pm}}(x,\!-v)) \!=\!\!  \lim_{t\to (t^{j}_{i,+})^\mp} \! Z^\eps_t$).

More precisely, we have 
	$$
	|Y_{t}^\eps (x,-v)-Y_{t}(x,-v)|
	\leq \eps \left[\frac{|Z|_{\max}}{\ka S_0}
	\ln \left(\frac{|Z|_{\max}}{|v|}\right)+2N \right]
	$$
	and
	$$
	|Z_{t}^\eps (x,-v)-Z_{t}(x,-v)|= \left|\sum_{j=1}^{N}\left[Z^\var_{\tau^{j}_{i,-}}-Z^\eps_{\tau^{j}_{i,+}}-\left(Z^+ _{t^{j}_{i,+}}-Z^-_{t^{j}_{i,+}}\right)\right]\right|
	=0
	$$

	Thanks to Lemma \ref{sovrapp} and the considerations given before, we may now study the behaviour of the limit points of  $\E[f_\eps(t,x,v) \indc_{A^\eps_1(x,v)}]$ when $\eps\to 0^+$ on $\R^d\times\R^d\setminus \{0\}$ (or equivalently the behaviour of $\E[f_\eps(t,x,v) \indc_{|v|>0} \indc_{A^\eps_1(x,v)}]$).

	\subsection{Limit of $\E[f_\eps(t,x,v)\indc_{|v|>0} \indc_{A^\eps_1}]$}

	We are now ready to prove Theorem \ref{TeorPrinc}, through the following auxiliary theorem.
	
	\begin{thm}
		\label{EqLim1}
		Assume $f_0\in C(\R^d\times\R^d)$, with compact support and $f_0\geq 0$, and let $f_\eps (t,x,v;C)$ be the solution to (\ref{evomisura}) with initial datum $f_0$
		Then, for $t\in(0,+\infty)$ on $\R^d\times\R^d\setminus\{0\}$, 
		
		$$
		\mu_\eps = \E[f_\eps (t,x,v)]dxdv \stackrel{*}{\deb} \mu=F(t,x,v)dxdv
		$$
		where $F$ solves the equation
		\be
		\label{eqlim}
		\pa_t F + v\cdot\grad_x F = \si\int_{0}^{1}dh^{d-1} \frac{|z|^d}{|v|^{d-1}}\frac{S(|z|)}{S(|v|)}F(t,x,z) \indc_{|v|>0}-\si |v| F(t,x,v) 
		\ee
	in $\cD'((0,+\infty)\times\R^d\times\R^d\setminus\{0\})$, with $z=a^{-1}\left(a(|v|)+2 \ka\sqrt{1-h^2}\right)\hat{v}$ and $\si=\lambda B^{d-1}$
	\end{thm}
	
	\begin{proof}
	Because of Lemma \ref{sovrapp}, 
	$$
	\E[f_\eps (t,x,v)\indc_{|v|>0}]=\E[f_\eps (t,x,v)\indc_{|v|>0}\indc_{A^\eps_1(x,v)}]+O(\eps).
	$$
	We write then, for $x\in Z_\eps[C]$, 
\begin{align}
\E [f_\eps & (t,x,v) \indc_{|v|>0}\indc_{A^\eps_1(x,v)}]
                 \nonumber \\
                 = & \sum_{N\geq 0}\frac{\l_{\eps}^N}{N!} \int dc_1\ldots dc_N \,
                 \indc_{|v|>0} \indc_{A^\eps_1(x,v)} \, e^{-\l_{\eps}|\mathcal{T}_\eps (-t;x,v)(C)|(c_1,\ldots,c_N)}
                  \nonumber \\
		&  \indc_{C\bigcap \mathcal{T}_\eps (-t;x,v)(C) 
		= (c_1, \ldots, c_N)} f_0((Y_{t}^\eps (x,\!-v), \! -Z_{t}^\eps(x,-v)))		
		J_{\eps,C}(t,x,v)
		\label{Asp}		
\end{align}
which, because of the considerations in Section \ref{markov}, can be rewritten as

\begin{eqnarray}
\lefteqn{\E[f_\eps (t,x,v)\indc_{|v|>0}\indc_{A^\eps_1}]=
		\sum_{N\geq 0}\l_{\eps}^N \int dc_1\ldots dc_N \indc_{[A^\eps_1(t,x,v)]_N}([C]_N)}
		\nonumber \\
		&&
		\times \exp\biggl\{-\lambda B^{d-1}\bigg[\sum_{j=0}^{N}
		|Z^\eps_{\tau^{j+1}_{i,+}}|(\tau^{j+1}_{i,+}-\tau^{j}_{i,-})+O(\eps)\bigg]\biggr\}
		\nonumber \\[2mm]
		&&
		\times f_0 \Bigg( \! x+\sum_{j=0}^{N}
		Z(\tau^{j+1}_{i,+})(\tau^{j+1}_{i,+}-\tau^{j}_{i,-})
		+O(\eps) \, ,
		\nonumber \\
		& & \hspace*{3cm} v+\hat{v}\sum_{j=1}^{N}
		\frac{\ka}{\eps}\int_{\tau^{j}_{i,+}}^{\tau^{j}_{i,-}}S(|Z|(s))|Z(s)| ds \Bigg)
		 \\
		&& 
		\times \left(\frac{\displaystyle |v|+\sum_{j=1}^{N}
			\frac{ \ka}{\eps}\int_{\tau^{j}_{i,+}}^{\tau^{j}_{i,-}}S(|Z|(s))|Z(s)|ds)}{|v|}\right)^{\!\!d}
			\nonumber \\[2mm]
		& & 
		\times \, \frac{ \displaystyle S\left(|v|+\sum_{j=1}^{N}
			\frac{\ka}{\eps}\int_{\tau^{j}_{i,+}}^{\tau^{j}_{i,-}}S(|Z|(s))|Z(s)|ds)\right)}{S(|v|)} \, .
		\nonumber
\end{eqnarray}
By the change of variables (\ref{cambvar}), and using the relations
$$
\label{prodotto}
	\frac{|Z^\eps_{\tau^{N}_{i,-}}|}{|v|}=\prod_{j=1}^{N} \; \frac{ \, \Big| Z^\eps_{\tau^{j}_{i,-}} \Big| \, }{ \Big|Z^\eps_{\tau^{j}_{i,+}} \Big|}
$$
	and
	$$ 
	\frac{ \, S(|Z^\eps_{\tau^{N}_{i,-}}|) \, }{S(|v|)}
	= \prod_{j=1}^{N} \, \frac{S(|Z^\eps_{\tau^{j}_{i,-}}|)}{S(|Z^\eps_{\tau^{j}_{i,+}}|)},
	$$
we get (recall that $\l_\eps=\frac{\l}{\eps^{d-1}}$)

\begin{eqnarray}
\lefteqn{\E[f_\eps (t,x,v)\indc_{|v|>0}\indc_{A^\eps_1}] }
	\nonumber \\
	&& =
	\sum_{N\geq 0}\l^N\int_{0}^{t}d\tau^{N}_{i,+}
	\int_{\nu_N^2\leq 1} d\nu_N \ldots \int_{0}^{\tau^{2}_{i,+}}d\tau^{1}_{i,+}
	\int_{\nu_1^2\leq 1} d\nu_1
	\nonumber \\
	&& \times \prod_{j=1}^{N}\left[|Z^\eps_{\tau^{j}_{i,+}}|
	\left(\frac{|Z^\eps_{\tau^{j}_{i,-}}|}{|Z^\eps_{\tau^{j}_{i,+}}|}\right)^d 
	\frac{S(|Z^\eps_{\tau^{j}_{i,-}}|)}{S(|Z^\eps_{\tau^{j}_{i,+}}|)}\right]
	\label{Asp2}  \\
	&& \times \exp\biggl\{-\lambda B^{d-1}\big[\sum_{j=0}^{N}
	|Z^\eps_{\tau^{j+1}_{i,+}}|(\tau^{j+1}_{i,+}-\tau^{j}_{i,-})+O(\eps)\big]\biggr\}
	\nonumber \\
	&&
	\times f_0 \Bigg(x+\sum_{j=0}^{N}
	Z(\tau^{j+1}_{i,+})(\tau^{j+1}_{i,+}-\tau^{j}_{i,-})
	+O(\eps) \, ,
	\nonumber \\
	&& \hspace*{2cm}
	v+\hat{v}\sum_{j=1}^{N}
	\frac{\ka}{\eps}\int_{\tau^{j}_{i,+}}^{\tau^{j}_{i,-}}S(|Z|(s))|Z(s)|ds) \Bigg)
	\nonumber
\end{eqnarray}
or equivalently
\begin{eqnarray}
\lefteqn{\E[f_\eps (t,x,v)\indc_{|v|>0}\indc_{A^\eps_1}]} 
	\nonumber \\
	&& =
	\sum_{N\geq 0}(\l B^{d-1})^N
	\int_{0}^{t}d\tau^{N}_{i,+}
	\int_{0}^1 dh^{d-1}_N \ldots \int_{0}^{\tau^{2}_{i,+}}d\tau^{1}_{i,+}
	\int_{0}^1 dh^{d-1}_1 
	\nonumber \\
	&& \times \prod_{j=1}^{N}\left[|Z^\eps_{\tau^{j}_{i,+}}|
	\left(\frac{|Z^\eps_{\tau^{j}_{i,-}}|}{|Z^\eps_{\tau^{j}_{i,+}}|}\right)^d 
	\frac{S(|Z^\eps_{\tau^{j}_{i,-}}|)}{S(|Z^\eps_{\tau^{j}_{i,+}}|)}\right]
	\label{Asp3} \\
	&& \times \exp\biggl\{-\lambda B^{d-1}\big[\sum_{j=0}^{N}
	|Z^\eps_{\tau^{j+1}_{i,+}}|(\tau^{j+1}_{i,+}-\tau^{j}_{i,-})+O(\eps) \big]\biggr\}
	\nonumber \\
	&&
	f_0 \Bigg( x-\sum_{j=0}^{N}
	Z(\tau^{j+1}_{i,+})(\tau^{j+1}_{i,+}-\tau^{j}_{i,-})
	+O(\eps) \, ,
	\nonumber \\
	&& \hspace*{2cm}
	v+\hat{v}\sum_{j=1}^{N}
	\frac{\ka}{\eps}\int_{\tau^{j}_{i,+}}^{\tau^{j}_{i,-}}S(|Z|(s))|Z(s)|ds \Bigg) .
	\nonumber
	\end{eqnarray}
We may now perform the limit, using (\ref{dinlim1})-(\ref{dinlim2}) and recalling that ${\bf P}(x\in B_\eps(c))<\eps$,
so that, for all $\phi \in C_c(\R^d\times\R^d\setminus\{0\})$,
we have
\begin{eqnarray*}
\lefteqn{ \lim_{\eps\to 0^+}\int_{\R^d\times\R^d}dxdv \, \phi(x,v) \, \E[f_\eps (t,x,v)]}
	\nonumber \\
        &=& \lim_{\eps\to 0^+}\int_{\R^d\times\R^d} \, dxdv \, \phi(x,v) \, \E[f_\eps (t,x,v)\indc_{|v|>0}\indc_{A^\eps_1}]
	\nonumber \\
        &=& \int_{\R^d\times\R^d}dxdv\phi(x,v)F(t,x,v),
\end{eqnarray*}
where 
			
\begin{eqnarray}
\lefteqn{ F(t,x,v)\!=\!
				\sum_{N\geq 0}(\l B^{d-1})^N \!\!
				\int_{0}^{t}d\tau^{N}_{i,+}
				\int_{0}^1 dh^{d-1}_N \ldots \int_{0}^{\tau^{2}_{i,+}}d\tau^{1}_{i,+}
				\int_{0}^1 dh^{d-1}_1 } 
				\nonumber \\
		&& \times \prod_{j=1}^{N}\left[|Z^-_{\tau^{j}_{i,+}}| \left(\frac{|Z^+_{\tau^{j}_{i,+}}|}{|Z^-_{\tau^{j}_{i,+}}|}\right)^d 
		      \frac{S(|Z^+_{\tau^{j}_{i,+}}|)}{S(|Z^-_{\tau^{j}_{i,+}}|)}\right]
		      \label{Asp4}  \\
		&& \times \exp\biggl\{-\lambda B^{d-1}\big[\sum_{j=0}^{N}
				|Z^-_{\tau^{j+1}_{i,+}}|(\tau^{j+1}_{i,+}-\tau^{j}_{i,+})\big]\biggr\}
				\nonumber \\
		&& \times f_0 \Bigg( x+\sum_{j=0}^{N}
				Z^-_{\tau^{j+1}_{i,+}}(\tau^{j+1}_{i,+}-\tau^{j}_{i,+}) \, ,
				    v+\sum_{j=1}^{N}\big[Z^+ _{\tau^{j}_{i,+}}-Z^-_{\tau^{j}_{i,+}} ] \Bigg)
				\nonumber
\end{eqnarray}
is the series solution to (\ref{eqlim}).
\end{proof}

\begin{proof}[Proof of Theorem \ref{TeorPrinc}]
	
	\phantom{Pippo}
	
	\smallskip
	
	\noindent
	We proved in Theorem \ref{EqLim1} that $\E[f_\eps(t,x,v;C)]dxdv\stackrel{*}{\deb}F(t,x,v)dxdv$ on $\R^d\times\R^d\setminus\{0\}$, where $F(t,x,v)$ solves (\ref{eqlim}) in $\cD'((0,+\infty)\times\R^d\times\R^d\setminus\{0\})$.
	 Therefore the limit point $F(t,dxdv)dt$ s.t. $F(t,dxdv)dt|_{(0,+\infty)\times\R^d\times\R^d\setminus\{0\}}=F(t,x,v)dtdxdv$ solves equation (\ref{equalimn0}) with 
	 $$m_F =\si\int_{0}^{1}dh^{d-1} \frac{|z|^d}{|v|^{d-1}}\frac{S(|z|)}{S(|v|)}F(t,x,z) \indc_{|v|>0} dxdv
	 $$ and 
	$$\ell(v) = \si|v|.$$
	
	Since the hypotheses 
	 in Lemma \ref{L-MassCons2} are verified, we have:
	\be
	\label{eqptolim3}
	(\pa_t+v\cdot\grad_x+\ell(v))F-m_F=\l_F(t,x)\de_0(dv)dx\,.
	\ee
 in $\cD'((0,+\infty)\times\R^d\times\R^d)$.
	
	Since the jacobian of the change of coordinate $z\to v$ is
	$$
	J'= \left(\frac{|z|}{|v|}\right)^{d-1}\frac{S(|z|)}{S(|v|)}
	$$
	and $|v|>0 \implies |z|>a^{-1}(2\ka\sqrt{1-h^2})$, $\l_F(t,x)$ is given by 
	\begin{align*}
	\l_F(t,dx) & =\int_{\R^d}\ell(v)F(t,dxdv) - \int_{\R^d}m_F(t,dxdv) 
	               \nonumber\\
	               & = \int_{\R^d} dz \int_{0}^{1} dh^{d-1}|z| F(t,x, z)\indc_{|z|\leq a^{-1}(2\ka\sqrt{1-h^2})} dx.
	\end{align*}	
	
	Since the solution to (\ref{eqlim2}) is unique by Lemma \ref{unic} and the weak-* limit of $\E[f_\eps(t,dxdv;C)]dt$ is also a weak limit, Theorem \ref{TeorPrinc} is proved.
	
\end{proof}

This concludes our description of the Boltzmann--Grad asymptotics for the particle dynamics described by (\ref{moto}).

\bigskip 

Finally, we have a few remarks.

The first one is that, although the dynamics (\ref{moto}) at $\eps>0$ induces an evolution on the Radon measures which preserves mass, in the limit mesoscopic evolution equation for the number density function (in the phase space) of the light particles the collision integral does not preserve mass. The conservation of mass is nevertheless guaranteed by the presence in the limit kinetic equation of a point measure component of the form $\l(t,x)\de_0(dv)dx$ proportional to the Dirac delta distribution in $0$, where $\l(t,x)$ is a sort of "Lagrange multiplier" of the conservation of mass, which describes the particles stopping in the random medium.

The second remark is that, because the collision integral does not satisfy the conservation of mass and the flow (\ref{moto}) is irreversible, although the dynamics of the particle system under exam looks quite simple, the asymptotics  behaviour of its number density profile is more tricky to analyse than for systems like the ones described, e.g. in \cite{Gallavotti} or \cite{DesvillettesRicci}. In particular, it is not possible to reduce the analysis to the asymptotics of the Markovian component of the number density (for a review on the terminology, see \cite{Ricci}).


\bigskip


  
  \section*{Declarations}
  
  Conflict of interest: the authors declare that  there are no conflicts of interest related to this article.

\section*{Data Availability Statement}

Data sharing is not applicable to this article as no datasets were generated or analyzed during the current study.

\

\


\begin{thebibliography}{0}
	
\bibitem{Duderstadt}
	Duderstadt, J.J., Moses, G.A.:
	\newblock Inertial Confinement Fusion.
	\newblock J. Wiley and Sons, Inc. (1982).
	
\bibitem{AtzeniMeyer}
	Atzeni, S., Meyer Ter-Vehn, J.:
	\newblock The Physics of Inertial Fusion.
	\newblock Clarendon Press-Oxford (2004).
	
\bibitem{Azechi}
	Azechi, H., Cable, M.D., Stapf, R.O.:
	\newblock Review of secondary and tertiary reactions, and neutron scattering as diagnostic techniques for inertial confinement fusion targets.
	\newblock Laser and Particle Beams \textbf{9},  119--134 (1993).
	\newblock \url{https://doi.org/10.1017/S0263034600002378}
	
\bibitem{Cable}
	Cable, M.D.,  Hatchett, S.P.:
	\newblock Neutron spectra from inertial confinement fusion targets for measurement of fuel areal density and charged particle stopping powers.
	\newblock J. Appl. Phys. \textbf{62}, 2233--2236 (1987). 
	\newblock \url{https://doi.org/10.1063/1.339850}	

\bibitem{Richt}
	Richtmyer, R.D.:
	\newblock Taylor Instability in Shock Acceleration of Compressible Fluids.
     \newblock 	Comm. on Pure and Appl. Math. \textbf{13}, 297--319  (1960).
     \newblock \url{https://doi.org/10.1002/cpa.3160130207}
     
\bibitem{LevermoreZimmerman} 
     Levermore, C.D., Zimmermann, G.B.:
     \newblock Modeling charged particle loss in a fuel/shell mixture.
     \newblock J. Math. Phys. {\bf 34}(10), 4725--4729 (1993).
     \newblock \url{https://doi.org/10.1063/1.530367}
	
\bibitem{LevermorePomraningSanzo}
	Levermore, C.D., Pomraning, G.C., Sanzo, D.L., Wong J.:
	\newblock Linear transport theory in a random medium.
	\newblock J. Math. Phys.{\bf27}(10), 2526--2536 (1986).
	\newblock \url{https://doi.org/10.1063/1.527320}
	
\bibitem{ClouetGolsePuelSentis} 
    Clouet, J.-F., Golse, F., Puel, M., Sentis, R.:
	\newblock On the slowing down of charged particles in a binary stochastic mixture.
	\newblock Kinetic and Related Models {\bf 1} (3), 387--404 (2008).
	\newblock \url{https://www.aimsciences.org/article/doi/10.3934/krm.2008.1.387}	

\bibitem{Ambrosio}
Ambrosio, L., Gigli, N., Savar\'e, G.:
\newblock Gradient Flows in Metric Spaces and in the Space of Probability Measures. 2nd ed.
\newblock Birkh\"auser Verlag AG, Basel, Boston, Berlin (2008)

\bibitem{Dieudonne2}
Dieudonn\'e, J.:
\newblock Treatise on Analysis. Vol. 2.
\newblock Academic Press, Inc. (1976)

\bibitem{Schwartz} 
Schwartz, L.:
\newblock Th\'eorie des distributions.
\newblock Hermann Paris (1966)

\bibitem{Malliavin}
Malliavin, P.: 
\newblock Integration and Probability.
\newblock Springer-Verlag New York, Inc. (1983)

\bibitem{Pazy}
Pazy, A.:
\newblock Semigroups of Operators and Applications to Partial Differential Equations.
\newblock Springer-Verlag New York, Inc. (1983)

\bibitem{Gallavotti}
Gallavotti, G.:
\newblock Rigorous theory of the Boltzmann equation in the Lorentz gas. 
\newblock Nota Interna No. 358,
Istituto di Fisica, Universit\`a di Roma (1972), reprint in Statistical Mechanics, 1999.
\newblock \url{https://ipparco.roma1.infn.it/pagine/deposito/1967-1979/041-Bz.pdf}


\bibitem{DesvillettesRicci}
Desvillettes, L., Ricci, V.: 
\newblock A Rigorous Derivation of a Linear Kinetic Equation of Fokker–Planck Type in the Limit of Grazing Collisions.
\newblock Journal of STatistical Physics {\bf 104} (5/6) 1173--1189 (2001).
\newblock \url{https://doi.org/10.1023/A:1010461929872}


\bibitem{Ricci}
Ricci, V.:  
\newblock Non Markovian Behavior of the Boltzmann-Grad
	Limit of Linear Stochastic Particle Systems.
\newblock Communications in Mathematical Sciences, Suppl. {\bf 1}, 95--105 (2007).
\newblock  \url{https://dx.doi.org/10.4310/CMS.2007.v5.n5.a8}
















\end{thebibliography}
\end{document}